\documentclass[english,a4paper,oneside,11pt]{amsart}
\usepackage{a4wide}
\usepackage{babel}

\usepackage{enumerate}

\usepackage[T1]{fontenc}
\usepackage[utf8]{inputenc}

\usepackage{amsmath}
\usepackage{amssymb}
\usepackage{amsthm}
\usepackage{xypic}
\usepackage{color}

\theoremstyle{plain}
\newtheorem{theorem}{Theorem}
\newtheorem*{theorem*}{Theorem}
\newtheorem{proposition}[theorem]{Proposition}
\newtheorem{corollary}[theorem]{Corollary}
\newtheorem*{corollary*}{Corollary}
\newtheorem{lemma}[theorem]{Lemma}

\theoremstyle{remark}
\newtheorem{remark}[theorem]{Remark}
\newtheorem{example}[theorem]{Example}

\theoremstyle{definition}
\newtheorem{definition}[theorem]{Definition}

\newtheorem{algorithm}[]{Algorithm}

\newcommand{\lend}[2]{\mathrm{End}(_{#1}#2)}
\newcommand{\rend}[2]{\mathrm{End}(#2_{#1})}
\newcommand{\raiz}{u}

\begin{document}

\title[Decoding RS Skew-differential Codes]{Decoding Reed-Solomon Skew-Differential Codes}

\author[G\'{o}mez-Torrecillas]{Jos\'{e} G\'{o}mez-Torrecillas}
\address{Department of Algebra, University of Granada}
\curraddr{}
\email{gomezj@ugr.es}

\author[Navarro]{Gabriel Navarro}
\address{CITIC and Department of Computer Science and Artificial Intelligence, University of Granada}
\email{gnavarro@ugr.es}

\author[S\'{a}nchez-Hern\'{a}ndez]{Jos\'e Patricio S\'anchez-Hern\'andez}
\address{Department of Algebra, University of Granada}
\email{jpsanchez@correo.ugr.es}
\thanks{Research supported by grants A-FQM-470-UGR18 from Junta de Andaluc\'{\i}a and FEDER and PID2019-110525GB-100 from AEI and FEDER}
\thanks{The third author was supported by The National Council of Science and Technology (CONACYT) with a scholarship  for a Postdoctoral Stay in the University of Granada.}

\date{\today}

\maketitle

\begin{abstract}
A large class of MDS linear codes is constructed. These codes are endowed with an efficient decoding algorithm. Both the definition of the codes and the design of their decoding algorithm only require from Linear Algebra methods, making them fully accesible for everyone. Thus, the first part of the paper develops a direct presentation of the codes by means of parity-check matrices, and the decoding algorithm rests upon matrix and linear maps manipulations.  The somewhat more sophisticated mathematical context (non-commutative rings) needed for the proof of the correctness of the decoding algorithm is postponed to the second part. A final section locates the Reed-Solomon skew-differential codes introduced  here within the general context of codes defined by means of skew polynomial rings. 
\end{abstract}

\section*{Introduction}
The treatment of cyclic linear codes as ideals of a quotient of a polynomial ring inspired the extension of cyclic-like conditions to the realm of skew-polynomial (non-commutative) rings both from the perspective of block codes \cite{Boucher/Geiselmann/Ulmer:2007, Boucher/Ulmer:2009, Boucher/Ulmer:2014, Martinez:2018} and convolutional codes \cite{Piret:1976,Roos:1979,Gluesing/Schmale:2004,Lopez/Szabo:2013,Gomez/alt:2016,Gomez/alt:2017a}. One of the nicest features of some (commutative) cyclic codes is the possibility of designing efficient algebraic decoding algorithms taking advantage of their rich algebraic structure \cite{Gorenstein/Zierler:1961,Peterson:1960,Sugiyama:1975}. These classic approaches have been adapted or, in some cases, inspire, decoding procedures for some families of cyclic-like codes based on non-commutative polynomial arithmetics \cite{Boucher/Ulmer:2014,Liu/etal:2015,Gomez/alt:2017b,Gomez/alt:2018a}. Dealing with these codes, presented in the language of left ideals and modules, requires a training in non-commutative rings which could limit their difusion and potential practical use among coding theorists and engineers. 

In this paper we simplify and extend to a considerably broader class of codes the algebraic decoding algorithms designed in \cite{Gomez/alt:2018a} and \cite{Gomez/alt:2019b} for skew RS and convolutional differential RS codes, respectively. These codes were presented as left ideals of certain non-commutative polynomial rings, and their decoding algorithms make use of advanced algebraic tools like evaluation of non-commutative polynomials. In contrast, both the construction of the codes in this note, and the description and implementation of their decoding algorithm only require basic Linear Algebra over a field. The introduction of the (minimal) algebraic machinery needed to prove the correctness of the decoding algorithm is postponed till Section \ref{fundamentos}. In this way, the material of Section \ref{Frontispicio} is ready to use even without knowing what a non-commutative ring is. 

We work over a general field $K$, since our interests include block linear codes ($K = \mathbb{F}_q$, a finite field) and convolutional codes ($K = \mathbb{F}_q(t)$, the rational function field in a variable $t$ over $\mathbb{F}_q$). 

We call our codes  Reed-Solomon (RS)  skew-differential codes because they are defined from a skew derivation of the field $K$ by means of parity check matrices (Definition \ref{Codefined}).  Each code $C_{(\varphi_u,\alpha,d)}$ depends on the three parameters reflected in the notation. Concretely, $\varphi_u$ is a transformation of $K$, defined from an element $u \in K$ and the skew derivation, that becomes linear with respect to a suitable subfield $K^{\varphi_u}$ of $K$, $\alpha$ is a cyclic vector of $\varphi_u$, and $d$ is the designed minimum Hamming distance of the code. After the description of the code $C_{(\varphi_u,\alpha,d)}$, Section \ref{Frontispicio} proceeds to the design of the its algebraic decoding algorithm (see Algorithm \ref{algoritmo}). It runs at follows: from a  word corrupted by up to $\tau = \lfloor \frac{d-1}{2} \rfloor$ errors, a matrix is computed recursively from the syndromes. The  left kernel of this matrix contains a nonzero vector $\rho$ (Proposition \ref{kernelrho}). With this vector at hand, a second matrix $L$ is recursively computed. This matrix gives an easy procedure to get the positions of the errors (Theorem \ref{Main}). Once these positions are known, the values of the errors are the solution of a linear system of equations. 

Section \ref{Ejemplos} deals with  the examples. We analyze when the finiteness condition (see Proposition \ref{phifinite}) that grants the construction of the code $C_{(\varphi_u,\alpha,d)}$ holds (Proposition \ref{genejemplos} and Corollary \ref{corejemplos}). Its turns out that in the cases of interest (block and convolutional skew-differential codes), the code $C_{(\varphi_u,\alpha,d)}$ is always built (Subsections \ref{block} and \ref{convolutional}). Concrete examples of application of Algorithm \ref{algoritmo} are shown in both cases. The computations are done with the help of the SageMath symbolic computation system \cite{sage}. 

Section \ref{fundamentos} is devoted to prove the mathematical results that ground Algorithm \ref{algoritmo}.  Our algebraic setup requires from the ring of additive endomorphisms of $K$ generated by $K$ and the map $\varphi_\raiz$, and the application of Jacobson-Bourbaki's Correspondence to identify which maps $\varphi_u$ are suitable for constructing skew-differential codes (Proposition \ref{JB}). Jacobson-Bourbaki's Correspondence is a more general version of Galois Correspondence that applies to additive maps (like $\varphi_\raiz$) which are not necessarily field automorphisms (see \cite{Sweedler:1975} or \cite{Winter:1974}). We also use the non-commutative Wronskian and Vandermonde matrices investigated in \cite{Delenclos/Leroy:2007}. Indeed, it turns out that $\mathcal{R}$ is isomorphic to the factor ring of a skew polynomial ring by the minimal polynomial of $\varphi_u$ (see Proposition \ref{isopoly}). This isomorphism eases the conceptual manipulation of the elements of $\mathcal{R}$ in Proposition \ref{Prop 3} and Theorem \ref{locatordiv}. The section is closed with the proof of the correctness of Algorithm \ref{algoritmo}.

Section \ref{SDmodule} gives a precise description of an RS skew-differential code $C_{(\varphi_\raiz,\alpha,d)}$ as a left ideal of $\mathcal{R}$ (Corollary \ref{ubicacion}). From the practical point of view, this serves to see in detail how the codes investigated in \cite{Gomez/alt:2018a} and \cite{Gomez/alt:2019b} are obtained as particular cases of $C_{(\varphi_\raiz, \alpha, d)}$ (see Examples \ref{RSauto} and \ref{RSder}). Besides, a generator  of $C_{(\varphi_\raiz,\alpha, d)}$ is explicitly given. On the theoretical side, Corollary \ref{ubicacion} identifies RS skew-differential codes as members of the very general family of module $(\sigma,\delta)$--codes defined in \cite{Boucher/Ulmer:2014}. In fact, we precisely characterize the module $(\sigma,\delta)$--codes with ``word ambient'' ring $\mathcal{R}$ (Proposition \ref{SDcodesversusmodule} and Proposition \ref{DCCdim}), and describe which of these codes are RS skew-differential codes (Corollary \ref{ubicacion}), thus enjoying an efficient algebraic decoding algorithm.

\section{Definition of the codes and specification of their algebraic decoding algorithm}\label{Frontispicio}
Let $K$ be a field. For any additive map\footnote{That is, it satisfies that $\phi(a+b) =\phi(a)+\phi(b)$ for all $a, b \in K$.} $\phi : K \to K$, set
\[
K^{\phi} = \{ b \in K : \phi(ab) = \phi(a)b \text{ for all } a \in K \}.
\]
A straightforward argument shows that $K^\phi$ is a subfield of $K$ and, obviously, $\phi$ becomes a $K^\phi$--linear map. A tempting idea is to use good enough field extensions $K/K^\phi$ to design $K$--linear error corrector codes with efficient algebraic decoding algorithms. In this note,  we consider additive maps on $K$ stemming from skew derivations. 

A \emph{skew derivation} on $K$ is a pair $(\sigma, \delta)$, where $\sigma$ is a field automorphism of $K$, and $\delta : K \to K$ is an additive map  subject to the condition
\begin{equation}\label{torcida}
\delta(ab) = \sigma(a)\delta(b) + \delta(a)b,
\end{equation}
for all $a, b \in K$. 

Given $\raiz \in K$, let $\varphi_u : K \to K$ be defined by
\begin{equation}
\varphi_u(a) = \sigma(a)\raiz + \delta(a),
\end{equation}
for all $a \in K$.

\begin{proposition}\label{phifinite}
Assume that the dimension of $K$ as a $K^{\varphi_u}$--vector space  is $m < \infty$. The minimal polynomial of the $K^{\varphi_u}$--linear map $\varphi_u$ has degree $m$ and, henceforth, it has at least a cyclic vector. Moreover $\alpha \in K$ is such a cyclic vector if and only if the matrix
\[
A = \left( \begin{array}{ccccc}
\alpha & \varphi_u(\alpha)& & \cdots &\varphi_u^{m-1}(\alpha)  \\
\varphi_u(\alpha)& \varphi_u^{2}(\alpha)&& \cdots & \varphi_u^{m}(\alpha)  \\

\vdots &\vdots && \ddots& \vdots\\
\varphi_u^{m-1}(\alpha)& \varphi_u^{m}(\alpha) && \cdots & \varphi_u^{2m-2}(\alpha)
\end{array} \right)
\]
is invertible. 
\end{proposition}

Proposition \ref{phifinite} establishes an adequate context to define our codes from the matrix $A$. It is worth to mention that the computation of a cyclic vector $\alpha$ can be randomized and does not require the computation of $K^{\varphi_u}$ (see Remark \ref{azar} in Section \ref{Ejemplos}).

\begin{definition}\label{Codefined}
Given $2 \leq d \leq m$, define the $K$--linear code $C_{(\varphi_u, \alpha, d)} \subseteq K^m$  of dimension $m-d+1$ as the left kernel of the matrix 
\[
H = \left( \begin{array}{ccccc}
\alpha & \varphi_u(\alpha)& & \cdots &\varphi_u^{d-2}(\alpha)  \\
\varphi_u(\alpha)& \varphi_u^{2}(\alpha)&& \cdots & \varphi_u^{d-1}(\alpha)  \\
\vdots &\vdots && \ddots& \vdots\\
\varphi_u^{m-1}(\alpha)& \varphi_u^{m}(\alpha) && \cdots & \varphi_u^{m+d-3}(\alpha)
\end{array} \right)
\]
\end{definition}

It can be proved (see Theorem \ref{MDS} below) that the minimum Hamming distance of this code is $d$, so it is an MDS code. We will call these codes \emph{Reed Solomon (RS) skew-differential codes.} It is worth to mention that when $K$ is either a finite field or a rational function field over a finite field, its dimension as a $K^{\varphi_{\raiz}}$--vector space is finite for any choice of the skew derivation $(\sigma,\delta)$, of the element $\raiz$ and of the cyclic vector $\alpha$ (see Section \ref{Ejemplos}). Indeed, $m$ is always equal to the order $m$ of the automorphism $\sigma$, when the latter is not the identity map.

Next, let us describe the decoding algorithm for $C_{(\varphi_u, \alpha, d)}$, that corrects up to $\tau = \lfloor \frac{d-1}{2} \rfloor$ errors.  Suppose that we receive a word $$y =(y_0, \dots, y_{m-1}) \in K^m$$ with   $ y= c + e \in K^m$, where $c$ is a codeword, and 
$$e = (e_0, \dots, e_{m-1})$$ is an error vector, which is assumed to be nonzero in the discussion below. Suppose that the nonzero components $e_{k_1}, \dots, e_{k_v} \in K$ of $e$ occur at the positions $0 \leq k_1 < \dots < k_v \leq m-1$. We assume that $v \leq \tau$. 

We start by computing, for $i = 0, \dots, d-2 $,  the syndromes
\begin{equation}\label{sindromes0}
S_{i,0} = \sum_{j=0}^{m-1} y_j\varphi_u^{i+j}(\alpha), 
\end{equation}
which are the components of the vector $yH$. 

For every pair $i, k$ of nonnegative integers such that $i+k \leq 2\tau-1$ we may compute $S_{i,k} \in K$ recursively from \eqref{sindromes0} according to the rule
\begin{equation}\label{sindromes1}
S_{i,k+1} = \sigma^{-1}(\delta(S_{i,k})-S_{i+1,k}).
\end{equation}
We may thus compute the columns of the matrix
\[
S = \left( \begin{array}{ccccc}
S_{0,0} & S_{0,1} && \cdots & S_{0,\tau - 1} \\
S_{1,0} & S_{1,1} && \cdots & S_{1,\tau - 1}  \\
\vdots &\vdots && \ddots& \vdots\\
S_{\tau,0} & S_{\tau,1} && \cdots &  S_{\tau,\tau - 1}
\end{array} \right).
\]

Next, for $1 \leq r \leq \tau$, let $S_r$ denote the matrix formed by the $r$ first columns of $S$ and compute
\[
\theta = \max \{ r  : rank \, S_r = r \}.
\]
\begin{proposition}\label{kernelrho}
The left kernel of the matrix
\[
B = \left( \begin{array}{ccccc}
S_{0,0} & S_{0,1} && \cdots & S_{0,\theta - 1} \\
S_{1,0} & S_{1,1} && \cdots & S_{1,\theta - 1}  \\
\vdots &\vdots && \ddots& \vdots\\
S_{\theta,0} & S_{\theta,1} && \cdots &  S_{\theta,\theta - 1}
\end{array} \right)
\]
is a one dimensional vector subspace of $K^{\theta + 1}$ spanned by a vector  $\rho = (\rho_0, \dots, \rho_\theta)$ with $\rho_\theta \neq 0$. 
\end{proposition}

The next step is the localization of the positions $k_1, \dots, k_v \in \{0, \dots, m-1\}$ at which the error values $e_{k_1}, \dots, e_{k_v}$ appear. This will be done with the help of a locator matrix built as follows. 

For $j = 0, \dots, m-1$  and $i = 0, \dots m-\theta-1$, set 
\begin{equation}\label{l0}
l_{0,j} = \begin{cases}
\rho_j & \text{if } j= 0, \dots, \theta \\
0 & \text{if } j = \theta + 1, \dots, m-1 \end{cases},  \qquad l_{i,-1} = 0.
 \end{equation}
We may then construct a matrix 
\begin{equation}\label{L}
L = \left( \begin{array}{ccccc}
l_{0,0} & l_{0,1} && \cdots & l_{0,m- 1} \\
l_{1,0} & l_{1,1} && \cdots & l_{1,m - 1}  \\
\vdots &\vdots && \ddots& \vdots\\
l_{m-\theta-1,0} & l_{m-\theta-1,1} && \cdots &  l_{m-\theta -1,m - 1}
\end{array} \right)
\end{equation}
by defining its entries recursively as
\begin{equation}\label{l1}
l_{i+1,j} = \sigma(l_{i,j-1}) + \delta(l_{i,j}).
\end{equation}

For $i= 0, \dots, m-1$ let $\epsilon_i$ denote the vector of $K^m$ whose $i$--th component equal to $1$, and every other component is $0$. By $Row(LA)$ we denote the row space of the matrix $LA$. 

\begin{theorem}\label{Main}
The error positions $k_1, \dots, k_v$ are, precisely, those $$k \in \{0,\dots,m-1\}$$ such that $\epsilon_k \notin Row(LA)$.  

The error values $e_{k_1}, \dots, e_{k_v} \in K$ are the unique solution of the linear system
\[
 S_{i,0} = \sum_{j=1}^v e_{k_j}\varphi_u^{i + k_j}(\alpha), \qquad (0 \leq i \leq v-1).
 \]
\end{theorem}

We are now ready to specify our decoding algorithm. 

\begin{algorithm}\label{algoritmo} \textsf{Decoding algorithm for an RS 
skew-differential code $C_{(\varphi_u,\alpha,d)}$.}

The input is a received word $y =(y_0, \dots, y_{m-1}) \in K^m$ with no more than  $\tau=\lfloor \frac{d-1}{2} \rfloor$ errors. 

The output is an error vector $e = (e_0, \dots, e_{m-1}) \in K^m$ such that $y-e \in C_{(\varphi_u, \alpha, d)}$. 

The algorithm runs according to the following steps:

\begin{enumerate}[1.] 
\item Compute $S_{i,0}$ according to \eqref{sindromes0} for $i=0,\dots, d-2$. If $S_{i,0} = 0$ for every $i = 0, \dots, d-2$, then $e = 0$.
\item Compute recursively $S_r$ for $r \geq 2$ by means of \eqref{sindromes1} until $rank \,  S_r < r$. Set $\theta = r-1$. 
\item Compute a nonzero $\rho = (\rho_0, \dots, \rho_\theta)$ in the kernel of the matrix $B$ formed by the first $\theta + 1$ rows of $S_\theta$.
\item Compute the matrix $L$ according to \eqref{l0} and \eqref{l1}.
\item The error positions set $T = \{ k_1, \dots, k_v \}$ is determined by
\[
T = \{ k \in \{0, \dots, m-1\} : \epsilon_k \notin Row(LA) \}.
\]
\item The error values $e_{k_1}, \dots, e_{k_v}$ are the solutions of the linear system 
\[ \sum_{j=0}^{m-1} y_j\varphi_u^{i+j}(\alpha) = \sum_{j=1}^v e_{k_j}\varphi_u^{i + k_j}(\alpha), \qquad (0 \leq i \leq v-1).
\]
\item Set $e_i = 0 $ for $i \notin T$. The error is $e=(e_0, \dots, e_{m-1})$. 
\end{enumerate}

\begin{remark}
The location of the error positions from the matrix $LA$ in the fifth step of Algorithm \ref{algoritmo} can be done by several methods. For instance, one may compute the reduced row echelon form of $LA$, as we do in the examples exhibited in Section \ref{Ejemplos}.
\end{remark}
\end{algorithm}

\section{Examples}\label{Ejemplos}
Before giving concrete examples, we discuss under which conditions the dimension of $K$ as a $K^{\varphi_\raiz}$--vector space is finite.  Of course, we also would like to avoid the extreme case where $K = K^{\varphi_u}$, so we also consider this situation. We first record separately an alternative description of $K^{\varphi_u}$, which will used several times. Keep the notation introduced in Section \ref{Frontispicio}.

\begin{lemma}\label{invconj}
The subfield $K^{\varphi_u}$ of $K$ admits the following description:
\[
K^{\varphi_u} = \{ b \in K \, | \, \sigma(b)\raiz + \delta(b) = \raiz b \}.
\]
\end{lemma}
\begin{proof}
It follows from the identities
\[
\varphi_u(ab) = \sigma(a)\sigma(b)\raiz + \sigma(a)\delta(b) + \delta(a)b
\]
and
\[
\varphi_u(a)b = \sigma(a)\raiz b + \delta(a)b,
\]
valid for any $a, b \in K$.
\end{proof}

\begin{proposition}\label{genejemplos}
\begin{enumerate}
\item\label{KKfi} The equality $K = K^{\varphi_u}$ holds if and only if $\delta(a) = \raiz(a - \sigma(a))$ for every $a \in K$.
\item\label{KKnofi} If $K \neq K^{\varphi_u}$ then
\[
K^{\varphi_u} = 
\begin{cases}
K^\delta & \text{ if }  \sigma = id_K \\
K^\sigma & \text{ if }  \sigma \neq id_K.
\end{cases}
\]
\end{enumerate}
\end{proposition}
\begin{proof}
Statement \eqref{KKfi} follows immediately  from Lemma \ref{KKfi}.

As for statement \eqref{KKnofi} concerns, let us first observe that, since $\delta(ab) = \delta(ba)$ for every $a, b \in K$, 
\begin{equation}\label{Lie}
\delta(b) (\sigma(a)-a) = \delta(a)(\sigma(b) -b),
\end{equation}
by virtue of \eqref{torcida}. Now, since $K \neq K^{\varphi_u}$, we pick $a \in K \setminus K^{\varphi_u}$. If $b \in K^{\varphi_u}$, then 
\[
\delta(a)(\sigma(b) - b) = \delta(b)(\sigma(a) - a) =(b - \sigma(b))\raiz(\sigma(a) -a),
\]
which is only possible if $\sigma(b) - b = 0$, as $\delta(a) \neq \raiz (\sigma(a) -a)$. Therefore, $b \in K^\sigma$ and, henceforth, $\delta(b) =\raiz(b-\sigma(b)) = 0$. We thus get that $K^{\varphi_u} \subseteq K^\sigma \cap K^\delta$. The converse inclusion is easily checked, so that we obtain
\begin{equation}\label{Kinter}
K^{\varphi_u} =  K^\sigma \cap K^\delta.
\end{equation}
If $\sigma = id_K$, then \eqref{Kinter} obviously implies $K^{\varphi_u} = K^\delta$. If $\sigma \neq id_K$, we may already pick $a \in K \setminus K^\sigma$. Then, for $b \in K^\sigma$ we get form \eqref{Lie}, that 
\[
\delta(b) (\sigma(a) - a) = \delta(a) (\sigma(b) - b) = 0.
\]
Hence, $\delta(b) = 0$ and $K^\sigma \subseteq K^\delta$, which implies, in view of \eqref{Kinter}, that $K^{\varphi_u} = K^\sigma$. 
\end{proof}

Given an automorphism $\sigma \neq id_K$ of the field $K$, and $v \in K$, we may define 
\begin{equation}\label{inder}
\delta_{\sigma, v} (a) = v(\sigma(a)-a), 
\end{equation}
for all $a \in K$, thus obtaining a map $\delta_{\sigma,v} : K \to K$ which is a $\sigma$--derivation. Indeed, it is already known (see, for instance, \cite[Proposition 1.1.20]{Jacobson:1996}), that every $\sigma$--derivation of the commutative field $K$ can be expressed in this form (this fact is easily derived from \eqref{Lie}). With this notation, we derive the following consequence of Proposition \ref{genejemplos}. 

\begin{corollary}\label{corejemplos}
Assume that $\sigma \neq id_K$ is an automorphism of $K$ of finite order $m$, and that $\delta = \delta_{\sigma,v}$. If $\raiz \neq - v$, then the dimension of $K$ over $K^{\varphi_\raiz}$ is $m$. 
\end{corollary}

\begin{remark}\label{azar}
In practice, the computation of one of the  cyclic vectors $\alpha$ predicted by Proposition \ref{phifinite} can be implemented by a randomized search in $K$ until the matrix $A$ becomes invertible. This avoids in most cases the computation of the subfield $K^{\varphi_\raiz}$, since the parameter $m$ is either the order of the automorphism $\sigma$ or, in the pure differential case of interest (namely $K = \mathbb{F}(t)$), the characteristic of the finite field $\mathbb{F}$.
\end{remark}

Next, we discuss how our construction applies to block and convolutional codes. We also illustrate the execution of our  decoding algorithm with some concrete examples.

\subsection{Block codes}\label{block}

Let us assume here that $K = \mathbb{F}$ is the finite field with $p^r$ elements for some prime $p$, so our codes become linear block codes over the alphabet $\mathbb{F}$. Every automorphism of $\mathbb{F}$ is a power of the Frobenius automorphism $\tau$ and, consequently, has finite order. Additionally, any derivation on $\mathbb{F}$ is inner, this is to mean, it is given by \eqref{inder}, so Corollary \ref{corejemplos} provides us plenty of non trivial examples. The steps of the design method of a RS skew-differential block code may be then enumerated as follows:
\begin{enumerate}
\item Choose a natural $0 < h < r$, and set $\sigma=\tau^h$ and $m=\frac{r}{(r,h)}$, the order of $\sigma$.
\item Choose $v$ and $u$ in $\mathbb{F}$, with $u +v\not = 0$, in order to set the $\sigma$-derivation $\delta:\mathbb{F}\to\mathbb{F}$ as $\delta(c)=v(\sigma(c)-c)$ and the additive map $\varphi_u$ as  $\varphi_u(c)=\sigma(c)u+\delta(c)$ for any $c\in \mathbb{F}$.
\item By a random search, find a cyclic vector $\alpha$ (see Remark \ref{azar}).
\item Finally, choose a designed distance $2\leq  d \leq m$, and set the parity check matrix $H$ as in Definition \ref{Codefined}.
\end{enumerate}
The degrees of freedom of this process suggest how wide this class of block codes is. Furthermore, RS skew-differential block codes are not cyclic, see Section \ref{SDmodule}. Nevertheless, Algorithm \ref{algoritmo} provides a decoding method as efficient as the classical Peterson-Gorenstein-Zierler algorithm.

Let us now describe a concrete example. Consider $\mathbb{F}=\mathbb{F}_2(a)$ the field with $256=2^8$ elements, where $a^{8} + a^{4} + a^{3} +a^{2} + 1=0$. For brevity, except for 0 and 1, we write the elements of  $\mathbb{F}$ as powers of $a$. Let  $\sigma$ be the Frobenius automorphism of $\mathbb{F}$, that is, $\sigma(c)=c^2$ for any $c\in \mathbb{F}$, which has order $m=8$.  Then, Corollary \ref{corejemplos} says that our code is of length $8$. We set $v=a$, yielding the $\sigma$-derivation given by $\delta(c)=ac^2+ac$ for every $c\in \mathbb{F}$, and $u=a^2$, so $\varphi_u (c) = a^{26}c^2 + ac$ for every $c\in \mathbb{F}$. 

We now choose $\alpha =  a^9$. The matrix $A$  from Proposition \ref{phifinite} takes now the form
\[A=\left(\begin{array}{llllllll}
a^{9} & a^{146} & a^{103} & a^{244} & a^{214} & a^{89} & a & a^{200} \\
a^{146} & a^{103} & a^{244} & a^{214} & a^{89} & a & a^{200} & a^{237} \\
a^{103} & a^{244} & a^{214} & a^{89} & a & a^{200} & a^{237} & a^{95} \\
a^{244} & a^{214} & a^{89} & a & a^{200} & a^{237} & a^{95} & a^{105} \\
a^{214} & a^{89} & a & a^{200} & a^{237} & a^{95} & a^{105} & a^{175} \\
a^{89} & a & a^{200} & a^{237} & a^{95} & a^{105} & a^{175} & a^{184} \\
a & a^{200} & a^{237} & a^{95} & a^{105} & a^{175} & a^{184} & a^{21} \\
a^{200} & a^{237} & a^{95} & a^{105} & a^{175} & a^{184} & a^{21} & a^{159}
\end{array}\right).\]
The determinant of $A$ equals $a^{47}$, so that $\alpha$ is a cyclic vector. Finally, we set a designed distance $d=5$. Let then $C = C_{(\varphi_u, a^9, 5)} \subseteq \mathbb{F}^8$ be the $[8,4,5]_{256}$-linear code defined as the left kernel of the following matrix $H$. From $H$, by standard methods, we have also computed a generating matrix $G$.
\[
H = \left(\begin{array}{llll}
a^{9} & a^{146} & a^{103} & a^{244}  \\
a^{146} & a^{103} & a^{244} & a^{214} \\
a^{103} & a^{244} & a^{214} & a^{89}   \\
a^{244} & a^{214} & a^{89} & a  \\
a^{214} & a^{89} & a & a^{200}  \\
a^{89} & a & a^{200} & a^{237}  \\
a & a^{200} & a^{237} & a^{95}  \\
a^{200} & a^{237} & a^{95} & a^{105}
\end{array}\right)  \text{ and }
G=\left(\begin{array}{llllllll}
1 & 0 & 0 & 0 & a^{105} & a^{69} & a^{221} & a^{41} \\
0 & 1 & 0 & 0 & a^{109} & a^{25} & a^{232} & a^{166} \\
0 & 0 & 1 & 0 & a^{145} & a^{54} & a^{104} & a^{36} \\
0 & 0 & 0 & 1 & a^{251} & a^{141} & a^{42} & a^{60}
\end{array}\right).
\]
The reader may refer to Section \ref{SDmodule} and Remark \ref{ncgenpol} for the explicit calculation of a non-commutative generator polynomial of $C$ so that the encoding can be performed in a similar way as for cyclic codes.

Let us exemplify the encoding-decoding process. We recall that the error-correcting capacity of $C$  is $\tau=2$. Suppose we want to transmit the message $$M = \left (a^{61}, a^{102}, a^{182}, a^{250} \right ),$$ so that we encode it to a codeword $$c = M G = \left (
a^{61}, a^{102}, a^{182}, a^{250}, a^{33}, a^{126}, a^{121}, a^{226}\right) \in C.$$
During the transmission, $c$ is corrupted by adding the error vector
\[e= \left(0, a^{2}, 0, a^{2}, 0, 0, 0, 0\right ),\]
yielding then the received word
\[y=c+e=
\left(a^{61}, a^{6}, a^{182}, a^{107}, a^{33}, a^{126}, a^{121}, a^{226}\right).\]
Now, we run Algorithm \ref{algoritmo}. We first calculate the syndromes
\[yH = \left ( a^{32}, a^{96}, a^{250}, a^{236} \right ),\]
so it is detected some error. The syndrome matrix is then
\[
S
=\left(\begin{array}{ll}
a^{32} & a^{3} \\
a^{96} & a^{67} \\
a^{250} & a^{221}
\end{array}\right).\]
The first column of $S$ is a multiple of its second column, so that $S$ has rank $1$ and, henceforth, $\theta = 1$.
Therefore, the matrix $B$  in Algorithm \ref{algoritmo} takes the form
\[B
=\left(\begin{array}{ll}
a^{32}  \\
a^{96} 
\end{array}\right).\]
and a basis of its left kernel is provided by the vector
\[\rho=
\left(a , a^{192}
\right).\]
The matrix $L$ defined in \eqref{L} becomes
\[L=\left(\begin{array}{llllllll}
a & a^{192} & 0 & 0 & 0 & 0 & 0 & 0 \\
a^{27} & a^{125} & a^{129} & 0 & 0 & 0 & 0 & 0 \\
a^{132} & a^{44} & a^{148} & a^{3} & 0 & 0 & 0 & 0 \\
a^{193} & a^{105} & a^{215} & a^{102} & a^{6} & 0 & 0 & 0 \\
a^{222} & a^{134} & a^{212} & a^{108} & a^{134} & a^{12} & 0 & 0 \\
a^{205} & a^{117} & a^{209} & a^{216} & a^{212} & a^{25} & a^{24} & 0 \\
a^{158} & a^{70} & a^{195} & a^{206} & a^{88} & a^{245} & a^{222} & a^{48}
\end{array}\right),\]
and $LA$ results
\[LA=\left(\begin{array}{llllllll}
a^{246} & a^{98} & a^{77} & a^{98} & a^{245} & a^{164} & a^{146} & a^{23} \\
a^{137} & a^{27} & a^{44} & a^{27} & a^{24} & a^{129} & a^{103} & a^{22} \\
a^{203} & a^{169} & a^{175} & a^{169} & a^{222} & a^{76} & a^{244} & a^{124} \\
a^{26} & a^{40} & a^{184} & a^{40} & a^{160} & a^{124} & a^{214} & a^{58} \\
a^{10} & a^{203} & a^{21} & a^{203} & a^{155} & a^{58} & a^{89} & a^{116} \\
a^{43} & a^{26} & a^{159} & a^{26} & a^{25} & a^{116} & a & a^{169} \\
a^{61} & a^{10} & a^{198} & a^{10} & a^{28} & a^{169} & a^{200} & a^{40}
\end{array}\right).\]
The identification of the positions $k \in \{0, 1 \dots, 7 \}$ such that $\epsilon_k \notin \mathrm{Row}(LA)$ can be easily done if we compute the row reduced echelon form of $LA$,
\[LA_{rref}=
\left(\begin{array}{llllllll}
1 & 0 & 0 & 0 & 0 & 0 & 0 & 0 \\
0 & 1 & 0 & 1 & 0 & 0 & 0 & 0 \\
0 & 0 & 1 & 0 & 0 & 0 & 0 & 0 \\
0 & 0 & 0 & 0 & 1 & 0 & 0 & 0 \\
0 & 0 & 0 & 0 & 0 & 1 & 0 & 0 \\
0 & 0 & 0 & 0 & 0 & 0 & 1 & 0 \\
0 & 0 & 0 & 0 & 0 & 0 & 0 & 1
\end{array}\right)\]
It is clear that $\epsilon_1$ and $\epsilon_3$ do not belong to $\mathrm{Row}(LA)$. Therefore, there are errors at positions 1 and 3. We finally need to solve a linear system in order to recover the error values. Indeed, the error
values are the solution of the system
\[\left(\begin{array}{rr}
a^{146} & a^{103} \\
a^{244} & a^{214}
\end{array}\right) \left(\begin{array}{r}
e_1\\
e_3
\end{array} \right)=\left(\begin{array}{rr}
a^{32} & a^{96}
\end{array}\right).\]
The solution is, as expected, $e_1=a^2$ and $e_3=a^2$.

\subsection{Convolutional codes}\label{convolutional}
Another case of interest is  $K = \mathbb{F}(t)$, the field of rational functions over a finite field $\mathbb{F}$.   Linear codes over $\mathbb{F}(t)$ are examples convolutional codes, see \cite{Forney:1970} for details. It is well-known that the group \(\mathrm{Aut}_{\mathbb{F}}{\mathbb{F}(t)}\) of all $\mathbb{F}$--linear automorphisms of the field $\mathbb{F}(t)$ can be identified with the projective general linear group \(\operatorname{PGL}(2,\mathbb{F})\) via the map $\Phi: \operatorname{PGL}(2,\mathbb{F}) \to \mathrm{Aut}_{\mathbb{F}}{\mathbb{F}(t)}$, which maps any matrix $M=\left (\begin{smallmatrix} \sigma_1 & \sigma_2 \\ \sigma_3 & \sigma_4 \end{smallmatrix}\right )\in \operatorname{PGL}(2,\mathbb{F})$ to the automorphism $\Phi(M)$ determined by the rule $t\mapsto \frac{\sigma_1 t + \sigma_2}{\sigma_3 t + \sigma_4}$. Every automorphism of $K$ has then finite order and,  on the other hand, the field of constants of any derivation of $\mathbb{F}(t)$ has finite index. Thus, Proposition \ref{genejemplos} says that virtually all choices of $\sigma$, $\delta$ and $u$ lead to non trivial RS skew-differential convolutional codes to which the decoding algorithm \ref{algoritmo} may be applied.

\begin{remark}
Algorithm \ref{algoritmo} deals with the Hamming metric, which is not the usual distance considered in convolutional codes. However, the use of Hamming distances in the convolutional setting might be of interest in the technology of distributed storage (see \cite[Sect. 2]{Gomez/alt:2019b}).
\end{remark}

Let us now detail a specific example. Let $\mathbb{F}=\mathbb{F}_2(a)$,  where $a^2+a+1=0$, the field with four elements and set $K=\mathbb{F}(t)$ the field of rational functions with coefficients in $\mathbb{F}$.   We shall follow likewise the construction method in Subection \ref{block}. 

As commented above, an automorphism of $K$ is determined by four elements $\sigma_1,\sigma_2,\sigma_3,\sigma_4$ in $\mathbb{F}$ verifying $\sigma_1\sigma_4-\sigma_2\sigma_3 \not = 0$. Set $\sigma_1=0$, $\sigma_2=1$, $\sigma_3 = 1$ and $\sigma_4 = a$ yielding the automorphism $\sigma:K\to K$ determined by $\sigma(t)=1/(t+a)$, which has order $m=5$. For simplicity, we fix $v=1$, so that $\delta(c)=\sigma(c)-c$ for any $c\in K$, and $u=0$, and then $\varphi_u = \delta$.  Now, consider $\alpha = t$. Since the matrix 
$$A=\left(\begin{array}{ccccc}
t & \frac{t^{2} + a t + 1}{t + a} & \frac{t^{2} + a t + 1}{ t + 1} & \frac{t^{4} + at^{3} + t^{2}}{ t^{3} + 1} & \frac{ t^{2} + a t + 1}{ t} \\ [0.8ex]
\frac{t^{2} + a t + 1}{t + a} & \frac{ t^{2} + a t + 1}{ t + 1} & \frac{ t^{4} + at^{3} +  t^{2}}{ t^{3} + 1} & \frac{ t^{2} + a t + 1}{t} & \frac{t^{2} + a t + 1}{a^2 t^{2} + t} \\ [0.8ex]
\frac{t^{2} + a t + 1}{ t + 1} & \frac{t^{4} + at^{3} +t^{2}}{ t^{3} + 1} & \frac{ t^{2} + a t + 1}{ t} & \frac{t^{2} + a t + 1}{a^2 t^{2} + t} & \frac{a^2 t^{4} + t^{3} + a t^{2} + a t + 1}{a^2 t^{3} + a t^{2} + t} \\ [0.8ex]
\frac{ t^{4} + at^{3} +  t^{2}}{ t^{3} +  1} & \frac{ t^{2} + at + 1}{ t} & \frac{t^{2} + a t + 1}{a^2 t^{2} + t} & \frac{a^2 t^{4} + t^{3} + a t^{2} + a t + 1}{a^2 t^{3} + a t^{2} + t} & \frac{ t^{2} +a t + 1}{a t^{2} +  t} \\ [0.8ex]
\frac{ t^{2} + a t + 1}{ t} & \frac{t^{2} + a t + 1}{a^2 t^{2} + t} & \frac{a^2 t^{4} + t^{3} + a t^{2} + a t + 1}{a^2 t^{3} + a t^{2} + t} & \frac{ t^{2} + a t + 1}{a t^{2} +  t} & \frac{t^{2} + a t + 1}{t + a + 1}
\end{array}\right)$$
is non-singular, we get from Proposition \ref{phifinite} that $\alpha$ is a cyclic vector for $\delta$.  Finally, the designed distance is selected to be $d=3$. So the skew-differential convolutional code $C=C_{(\delta,t,3)}$ can correct a single error, and a parity check matrix takes the form
$$
H=\left(\begin{array}{cc}
t & \frac{t^{2} + a t + 1}{t + a} \\ [0.8ex]
\frac{t^{2} + a t + 1}{t + a} & \frac{t^{2} + a t + 1}{t + 1} \\ [0.8ex]
\frac{t^{2} + a t + 1}{ t + 1} & \frac{ t^{4} +a^2 t^{3} +  t^{2}}{ t^{3} + 1} \\ [0.8ex]
\frac{ t^{4} +a^2 t^{3} + t^{2}}{ t^{3} +1} & \frac{ t^{2} + a t + 1}{ t} \\ [0.8ex]
\frac{ t^{2} + a t + 1}{ t} & \frac{t^{2} + a t + 1}{a^2 t^{2} + t}
\end{array}\right)
$$
Let us briefly exemplify our decoding algorithm. Suppose that we receive the word
$$y=\left(0,\,1,\,a^2,\,\frac{ t^{2} +  t}{a^2t^{2} +  t + 1},\,0\right),$$
whose matrix of syndromes is as follows:
$$S=\left(\begin{array}{c}
\frac{t^{3} + a t^{2} + t}{t^{4} + a t^{2} + a t + 1} \\ [0.8ex]
\frac{ t^{3} +a t^{2} + t}{a^2 t^{5} + t^{4} + t^{3} + a^2 t^{2} + t +  1}
\end{array}\right).$$
Henceforth, the system detects errors during the transmission. Clearly $\theta= 1$ and $B=S$, and the vector $\rho$ becomes $\rho=\left(1,at + 1\right)$. The matrix $L$ takes the form
$$
L=\left(\begin{array}{ccccc}
1 & a^2 t + 1 & 0 & 0 & 0 \\ [0.8ex]
0 & \frac{a^2 t^{2} + 1}{t + a} & \frac{t + 1}{t + a} & 0 & 0 \\ [0.8ex]
0 & \frac{t^{2} + a t + 1}{a t + a} & 1 & \frac{t + a + 1}{a t + a} & 0 \\ [0.8ex]
0 & \frac{ t^{4} + a t^{3} +  t^{2}}{a t^{3} + a } & \frac{t^{2} + a t + 1}{a t^{2} + a t + a} & \frac{1}{a^2 t^{2} + t + a} & \frac{a^2 t}{a t + 1}
\end{array}\right).
$$
We then compute $LA$ and its row reduced echelon form obtaining that
$$LA_{\mathrm{rref}}=
\left(\begin{array}{ccccc}
1 & 0 & 0 & 0 & 0 \\
0 & 1 & 0 & 0 & 0 \\
0 & 0 & 1 & 0 & 0 \\
0 & 0 & 0 & 1 & 0
\end{array}\right).$$
It is clear that $\epsilon_4$ does not belong to $\mathrm{Row}LA$, so we find an error at position 4. When computing the error value we find that
$$e_4=\frac{t^{2}}{t^{4} + at^{2} + a t + 1}.$$
Therefore, the correction gives the codeword
$$c=\left(0,\,1,\,a^2,\,\frac{ t^{2} +  t}{a^2t^{2} +  t + 1},\,\frac{t^{2}}{t^{4} + at^{2} + a t + 1}\right),$$
and the original message would be $M=\left (0,\,1,\,a^2 \right )$.

\section{Mathematical set up and proofs}\label{fundamentos}

The aim of this section is to prove the mathematical results that ground Algorithm \ref{algoritmo}. 
So, let $(\sigma,\delta)$ be a skew-derivation on a field $K$, as defined in Section \ref{Frontispicio}. Recall that, for each $\raiz \in K$, we define
\begin{equation}\label{varphi}
\varphi_u (a) = \sigma(a)\raiz + \delta(a),
\end{equation}
for all $a \in K$, thus obtaining a map $\varphi_u: K \to K$. This additive map becomes right $K^{\varphi_u}$--linear, where 
\[
K^{\varphi_u} =\{  b \in K : \varphi_u(ab) = \varphi_u(a)b \text{ for all } a \in K\}
\]
is the $\varphi_u$--invariant subfield of $K$. 

Let $\lend{}{K}$ denote the ring of endomorphisms of $K$  as an additive group.  Let $\mathcal{R}$ be the subring of $\lend{}{K}$ generated by $K$ and $\varphi_u$. Here, $K$ is seen as a subring of $\lend{}{K}$ by considering each element $a$ of $K$ as the additive endomorphism given by multiplication  by $a$.

\begin{proposition}\label{JB}
 If the dimension of $K$ as a $K^{\varphi_u}$--vector space is $m < \infty$, then the minimal polynomial of $\varphi_u$ as a $K^{\varphi_u}$--linear map has degree $m$. Consequently, $\varphi_u$ has at least a cyclic vector $\alpha \in K$. Moreover, 
\begin{equation}\label{Kplus}
\mathcal{R} = K \oplus K\varphi_u \oplus \cdots \oplus K\varphi_u^{m-1}.
\end{equation}
\end{proposition}
\begin{proof}
It easily follows from \eqref{torcida} that, in $\lend{}{K}$,
\begin{equation}\label{phicommutation}
\varphi_u a = \sigma(a) \varphi_u + \delta(a),
\end{equation}
for all $a \in K$. This implies that $\mathcal{R} = K + K \varphi_u + K \varphi_u^2 + \cdots $. 

Now, since $\dim_{K^{\varphi_\raiz}}K = m$, the minimal polynomial of $\varphi_u$ as a $K^{\varphi_\raiz}$--linear map has degree $n \leq m$. This in particular implies that $\mathcal{R} = K + K\varphi_u + \cdots + K\varphi_u^{n-1}$.  Oh the other hand, by Jacobson-Bourbaki's correspondence  \cite[Theorem 4.1]{Sweedler:1975}, $m = \dim_{K}\mathcal{R}$. We thus derive that $n = m$ and \eqref{Kplus}.
\end{proof}

In the rest of the paper, we assume that $\dim_{K^{\varphi_u}} K = m < \infty$. According to Proposition \ref{JB}, the minimal equation of $\varphi_u$ over $K^{\varphi_u}$ has degree $m$, that is, is of the form 
\begin{equation}\label{mucero}
 0 = \varphi_u^{m} + \mu_{m-1}\varphi_u^{m-1} + \cdots + \mu_1\varphi_u + \mu_0
 \end{equation}
with $\mu_i \in K^{\varphi_u}$ for $i= 0, \dots, m-1$.

Let $\alpha \in K$. For any subset $\{t_1, \dots, t_n \} \subseteq \{0,\dots,m-1\}$, define, following \cite{Delenclos/Leroy:2007}, the matrix
\[
W(\varphi_u^{t_1}(\alpha), \dots, \varphi_u^{t_n}(\alpha)) = 
\left( \begin{array}{ccccc}
\varphi_u^{t_1}(\alpha) &\varphi_u^{t_2}(\alpha)&& \cdots & \varphi_u^{t_n}(\alpha)\\
\varphi_u^{t_1+1}(\alpha) & \varphi_u^{t_2+1}(\alpha) && \cdots &  \varphi_u^{t_n+1}(\alpha)  \\
\vdots &\vdots && \ddots& \vdots\\
\varphi_u^{t_1+n-1}(\alpha)&  \varphi_u^{t_2+n-1}(\alpha)&& \cdots &  \varphi_u^{t_n+n-1}(\alpha)
\end{array} \right). 
\]

\begin{lemma}\label{Wronskian}
Given $\alpha \in K$, the following conditions are equivalent.
\begin{enumerate}
\item $\alpha$ is a cyclic vector for the $K^{\varphi_u}$--linear map $\varphi_u$.
\item $W(\alpha, \varphi_u(\alpha), \dots, \varphi_u^{m-1}(\alpha))$ is an invertible matrix. 
\item $W(\varphi_u^{t_1}(\alpha), \dots, \varphi_u^{t_n}(\alpha))$ is an invertible matrix for every subset $\{t_1, \dots, t_n \} \subseteq \{0,\dots,m-1\}$.
\end{enumerate}
\end{lemma}
\begin{proof}
For every nonzero $c \in K$, consider the conjugate of $\raiz$ by $c$:
 \[ {}^c\raiz = \sigma(c)\raiz c^{-1} + \delta(c) c^{-1}.\] 
 By Lemma \ref{invconj},
 \[
 K^{\varphi_u} =  \{ c \in K\setminus \{ 0 \} \; | \; {}^c\raiz = \raiz \} \cup \{ 0 \};
\]
the latter being the $(\sigma-\delta)$--centralizer of $\raiz$ in the terminology of \cite{Delenclos/Leroy:2007}. Since $\alpha$ is a cyclic vector for $\varphi_u$ precisely when
$
\{\alpha, \varphi_u(\alpha), \dots, \varphi_u^{m-1}(\alpha) \}$ is a $K^{\varphi_u}$--basis of $K$, we may apply  \cite[Theorem 5.3]{Delenclos/Leroy:2007} to deduce that the three  conditions  are equivalent. 
 \end{proof}
 
\noindent
 \textbf{Proof of Proposition \ref{phifinite}.}
 It is a  consequence of Proposition \ref{JB} and Lemma \ref{Wronskian}.
 
 \medskip

Fix a cyclic vector $\alpha \in K$ of $\varphi_u$.  Let $A = W(\alpha, \varphi_u(\alpha), \dots, \varphi_u^{m-1}(\alpha))$ which, by Lemma \ref{Wronskian}, is an invertible matrix with coefficients in $K$. 

\begin{theorem}\label{MDS}
For $2 \leq d \leq m$, let $C_{(\varphi_u, \alpha, d)} \subseteq K^m$ be the left kernel of the matrix 
\begin{equation}\label{Hache}
H = \left( \begin{array}{ccccc}
\alpha & \varphi_u(\alpha)& & \cdots &\varphi_u^{d-2}(\alpha)  \\
\varphi_u(\alpha)& \varphi_u^{2}(\alpha)&& \cdots & \varphi_u^{d-1}(\alpha)  \\
\vdots &\vdots && \ddots& \vdots\\
\varphi_u^{m-1}(\alpha)& \varphi_u^{m}(\alpha) && \cdots & \varphi_u^{m+d-3}(\alpha)
\end{array} \right).
\end{equation}
Then $C_{(\varphi_u,\alpha,d)}$ is a $K$--linear code of dimension $m-d+1$ and minimum Hamming distance $d$.
\end{theorem}
\begin{proof}
Since $H$ consists of the first $d-1$ columns of the invertible matrix $A$, we get that the dimension of the left $K$--vector subspace $C_{(\varphi_u,\alpha,d)}$ is $m-d+1$.  Every submatrix  $M$  of order $d-1$  of $H$ is of the form
$$M = \left( \begin{array}{ccccc}
\varphi_u^{k_1}(\alpha) & \varphi_u^{k_1+1}(\alpha)& & \cdots &\varphi_u^{k_1 + d-2}(\alpha)  \\
\varphi_u^{k_2}(\alpha)& \varphi_u^{k_2+1}(\alpha)&& \cdots & \varphi_u^{k_2+ d-2}(\alpha)  \\

\vdots &\vdots && \ddots& \vdots\\
\varphi_u^{k_{d-1}}(\alpha)& \varphi_u^{k_{d-1}+1}(\alpha) && \cdots & \varphi_u^{k_{d-1} +d-2}(\alpha)
\end{array} \right), $$
 where $\{k_1,\dots, k_{d-1}\}\subseteq\{0,\dots,m-1\}$. We see that
 \[
 M = W(\varphi_u^{k_1}{(\alpha)},\dots,\varphi_u^{k_{d-1}}{(\alpha}))^{t},
  \]
 which is, by Lemma \ref{Wronskian}, invertible. Hence, the Hamming distance of $C_{(\varphi_u,\alpha,d)}$ is $d$. \end{proof}

The proof of Lemma \ref{Wronskian} is based on a result from \cite{Delenclos/Leroy:2007} in the realm of the theory of skew polynomials. Indeed, for some purposes, it is useful to understand the ring $\mathcal{R}$ as a factor ring of a ring of skew polynomials. Let us derive such a description.

The skew derivation $(\sigma,\delta)$ leads to the construction of a non commutative polynomial ring $R = K[x;\sigma,\delta]$, often called a skew polynomial ring (see, e.g., \cite{Jacobson:1996}). The elements of $R$ are polynomials in an indeterminate $x$ with coefficients from $K$ written on the left (that is, the monomials $1, x, x^2, \dots $ form a basis of $R$ as a left vector space over $K$). The multiplication of $R$ is subject to the following rule:
\begin{equation}\label{commutacion}
xa = \sigma(a)x + \delta(a),
\end{equation}
for all $a \in K$. 

\begin{proposition}\label{isopoly}
The map $\pi : R  \to \mathcal{R}$ that sends $\sum_if_ix^i$ onto $\sum_if_i\varphi_u^i$ is a surjective ring homomorphism whose kernel is $R\mu = \mu R$,
where   $$\mu = x^{m} + \sum_{i= 0}^{m-1}\mu_ix^i$$
is a polynomial in $R$ built from the coefficients of the minimal equation of $\varphi_u$, see \eqref{mucero}.

Hence, there is a $K$--linear isomorphism of rings $R/R\mu \cong \mathcal{R}$.
\end{proposition}
\begin{proof}
Observe that $\pi$ is clearly left $K$--linear and, from Proposition \ref{JB}, surjective. It is multiplicative since $\varphi_u$ satisfies \eqref{phicommutation}. Its kernel is an ideal $I$ of $R$ which, as a left ideal, is generated by the monic polynomial $h \in R$ in $I$ of least degree, due to the left Euclidean division algorithm enjoyed by $R$ (see, e.g. \cite{Jacobson:1996}). Also, the degree of $h$ is the dimension of $R/I \cong \mathcal{R}$ as a left $K$--vector space. By Proposition \ref{JB}, this dimension equals $m$. We see that $\mu$ fits these requirements, so that $h = \mu$, and $I = R\mu$. Finally, since $I$ is an ideal, we get that $I = \mu R$ as well.  
\end{proof}

 We  may thus identify $\mathcal{R}$ with $R/R\mu$, and, therefore, its elements with polynomials in $R$ with degree smaller than $m$ (this identification makes correspond $\varphi_u$ with $x$). This view makes more natural some concepts, like the degree of an element of $\mathcal{R}$. 
 
 The  coordinate isomorphism of left $K$--vector spaces 
 \[
 \mathfrak{v} : \mathcal{R} \to K^m, \qquad (\sum_{i=0}^{m-1}f_ix^i \mapsto (f_0,f_1,\dots,f_{m-1}))
 \]
allows the transfer of elements and vector subspaces between both $K$--vector spaces. 

We are ready to consider our decoding algorithm. Let $c\in C_{(\varphi_u,\alpha,d)}$ be a codeword that is transmitted through a noisy channel, and let 
\[
y= (y_0,y_1, \dots, y_{m-1}) \in K^m
\]
be the received word. We may decompose $y = c + e$, where
\[
e=(e_0,e_1,\dots,e_{m-1}) \in K^m
\]
 is the error vector. By $k_1, \dots, k_v \in \{0, 1, \dots, m-1\}$ we denote the positions where the nonzero error values $e_{k_1}, \dots, e_{k_v} \in K$ occur. We prove first that the latter can be computed from $y$ once the positions are known.

 \begin{proposition}\label{errorsol}
 If $0 \leq i \leq d-2$, then
 \begin{equation}\label{errorsol1}
 \sum_{j= 0}^{m-1}y_j\varphi_u^{i+j}(\alpha) = \sum_{j=1}^v e_{k_j}\varphi_u^{i + k_j}(\alpha).
 \end{equation}
 Therefore, if $v \leq d-1$, then $(e_{k_1}, \dots, e_{k_v})$ is the unique solution of the linear system of equations
 \begin{equation}\label{errorsol2}
 \sum_{j=0}^{m-1} y_j\varphi_u^{i+j}(\alpha) = \sum_{j=1}^v e_{k_j}\varphi_u^{i + k_j}(\alpha), \qquad (0 \leq i \leq v-1).
 \end{equation}
 \end{proposition}
 \begin{proof}
 The equations \eqref{errorsol1} hold because $C_{(\varphi_u,\alpha,d)}$ is the left kernel of the matrix $H$ defined in \eqref{Hache}.  The linear system  \eqref{errorsol2} has a unique solution since the matrix
 $$\left( \begin{array}{ccccc}
\varphi_u^{k_1}(\alpha) &\varphi_u^{k_1+1}(\alpha)  && \cdots &\varphi_u^{k_1+v-1}(\alpha)  \\

\varphi_u^{k_2}(\alpha) & \varphi_u^{k_2+1}(\alpha) && \cdots &  \varphi_u^{k_2+v-1}(\alpha) \\

\vdots &\vdots && \ddots& \vdots\\

\varphi_u^{k_v}(\alpha) & \varphi_u^{k_v+1}(\alpha) && \cdots &  \varphi_u^{k_v+v-1}(\alpha)
\end{array} \right) = W(\varphi_u^{k_1}(\alpha), \dots, \varphi_u^{k_v}(\alpha))^{t}$$  
 is invertible by Lemma \ref{Wronskian}. 

 \end{proof}

Our  aim is then to design an algorithm for computing the positions $k_1, \dots, k_v$ where the errors $e_{k_1}, \dots, e_{k_v}$ appear. We assume in our exposition that $e \neq 0$.  

\medskip

For every pair $(i,k)$ of non-negative integers,  set  
\begin{equation}\label{Sik}
S_{i,k}=\sum_{j=1}^v\varphi_u^{i+k_j}(\alpha)\psi^k(e_{k_j}),
\end{equation}
where 
\begin{equation}\label{psi}
\psi(a) = \sigma^{-1}(\delta(a) - \raiz a)
\end{equation}
 for all $a \in K$. 

\begin{lemma}\label{calcular}
For all pairs $(i,k)$ of non-negative integers, we have
\begin{equation}\label{Sik+1}
\sigma(S_{i,k+1}) = \delta(S_{i,k})-S_{i+1,k}
\end{equation}
Moreover, 

\begin{equation}\label{S1}
S_{i,0} =  \sum_{j= 0}^{m-1}y_j\varphi_u^{i+j}(\alpha), 
\end{equation}
for every $i=0, \dots, d-2$, and  the values $S_{i,k}$ can be computed recursively by means of  \eqref{Sik+1} from the received word $y$ whenever $i + k \leq d-2$.
\end{lemma}
\begin{proof}
Observe that
\begin{equation}\label{psiphi}
\sigma(a\psi(b)) = \delta(ab) - \varphi_u(a)b,
\end{equation}
for all $a,b \in K$.
Indeed, 
\[
\begin{array}{rcl}
\sigma(a\psi(b)) & \stackrel{\eqref{psi}}{=} & \sigma(a)(\delta(b)-\raiz b) \\
& \stackrel{\eqref{torcida}}{=} & \delta(ab) - \delta(a)b - \sigma(a)\raiz b \\
& \stackrel{\eqref{varphi}}{=} & \delta(ab) - \varphi_u(a)b.
\end{array}
\]
 
For every pair $(i,k)$, 

\[
\begin{array}{rcl}
\sigma(S_{i,k+1})&\stackrel{\eqref{Sik}}{=}& \sum_{j=1}^v \sigma(\varphi_u^{i+k_j}(\alpha) \psi^{k+1}(e_{k_j}) )\\
& \stackrel{\eqref{psiphi}}{=} & \sum_{j=1}^v \delta(\varphi_u^{i+k_j}(\alpha)\psi^k(e_{k_j})) - \sum_{j=1}^v \varphi_u^{i+ k_j + 1}(\alpha) \psi^k(e_{k_j}) \\
& \stackrel{\eqref{Sik}}{=} & \delta(S_{i,k}) - S_{i+1,k}.
\end{array}
\] 

Finally, since $K$ is commutative,  \eqref{S1} follows from \eqref{errorsol1}. 
\end{proof}

Set $T = \{k_1, \dots, k_v\}$, and let $A_T$ be the submatrix of $A = W(\alpha, \varphi_u(\alpha), \dots, \varphi_u^{m-1}(\alpha))$ formed by the columns at positions $k_1, \dots, k_v$, that is
\[
A_T = \left( \begin{array}{ccccc}
\varphi_u^{k_1}(\alpha) &\varphi_u^{k_2}(\alpha)  && \cdots &\varphi_u^{k_v}(\alpha)  \\

\varphi_u^{k_1+1}(\alpha) & \varphi_u^{k_2+1}(\alpha) && \cdots &  \varphi_u^{k_v+1}(\alpha) \\

\vdots &\vdots && \ddots& \vdots\\

\varphi_u^{k_1+m-1}(\alpha) & \varphi_u^{k_2+m-1}(\alpha) && \cdots &  \varphi_u^{k_v+m-1}(\alpha)
\end{array} \right).
\]

\begin{proposition}\label{Prop 3} Define, for every $1 \leq r$, the matrix
\[ E_r=\left( \begin{array}{ccccc}
e_{k_1} &\psi(e_{k_1})  && \cdots &\psi^{r-1}(e_{k_1})  \\
e_{k_2} & \psi(e_{k_2}) && \cdots &  \psi^{r-1}(e_{k_2}) \\
\vdots &\vdots && \ddots& \vdots\\
e_{k_v} & \psi(e_{k_v} ) && \cdots &  \psi^{r-1}(e_{k_v} )
\end{array} \right).
\]

and set 
\[
\theta = \max \{ r : rank \, E_r = r  \}.
\]
\begin{enumerate}
\item\label{ro} If $V \subseteq K^m$ is the left kernel of the matrix $A_T E_{\theta}$, then  $\mathfrak{v}^{-1}(V)=\mathcal{R}\rho$ for some  $\rho\in \mathcal{R}$ of degree $\theta$. 
\item\label{rocalculo} If $B$ is the matrix formed by the first $\theta + 1$ rows of $A_TE_{\theta}$, then we may choose $\rho = \rho_0 + \rho_1x + \cdots + \rho_{\theta}x^\theta$, for any nonzero vector $(\rho_0,\rho_1, \dots, \rho_\theta)$ in the left kernel of $B$. 
\end{enumerate}
\end{proposition}
\begin{proof}
\eqref{ro}  
We will prove that the $K$--vector subspace $I = \mathfrak{v}^{-1}(V)$ of $\mathcal{R}$ is a left ideal. To do this, we need just to check that $xI \subseteq I$.  
Given $\sum_{i=0}^{m-1}a_ix^i \in \mathcal{R}$ we get from \eqref{commutacion}, since $\mu = 0$ in $\mathcal{R}$,  that 
\begin{equation}\label{fea}
x(\sum_{i= 0}^{m-1}a_ix^i) = \sum_{i=0}^{m-1}(\sigma(a_{i-1}) + \delta(a_i) - \sigma(a_{m-1})\mu_i)x^i,
\end{equation}
 where we set $a_{-1} = 0$.

Suppose that $(a_0,\dots,a_{m-2},a_{m-1})A_T E_{\theta}=0$.  The maximality of $\theta$ ensures that the last column of $E_{\theta+1}$ is a right linear combination of the former $\theta$ columns. Hence, $$(a_0,\dots,a_{m-2},a_{m-1})A_T E_{\theta+1}=0.$$

Observe that \[ 
A_T E_{\theta +1 }=  
\left( \begin{array}{ccccc}
S_{0,0} & S_{0,1} && \cdots & S_{0,\theta} \\
S_{1,0} & S_{1,1} && \cdots & S_{1,\theta}  \\
\vdots &\vdots && \ddots& \vdots\\
S_{m-1,0} & S_{m-1,1} && \cdots &  S_{m-1,\theta}
\end{array} \right).
 \]
Therefore, 
\begin{equation} \label{sumas}
\sum_{i=0}^{m-1}a_iS_{i,k}=0, \qquad \text{ for all } 0 \leq k \leq \theta.
\end{equation}

For $0\leq k \leq \theta-1$ we have

\renewcommand{\arraystretch}{1.5}
$$\begin{array}{rcll}
\sum_{i=0}^{m-1}(\sigma(a_{i-1})+\delta(a_i))S_{i,k} &\stackrel{\eqref{torcida}}{=} &\sum_{i=0}^{m-1}\{\sigma(a_{i-1})S_{i,k}+\delta(a_iS_{i,k})-\sigma(a_{i})\delta(S_{i,k})\} \\&\stackrel{\eqref{sumas}}{=} & \sum_{i=0}^{m-1}\sigma(a_{i-1})S_{i,k}- \sum_{i=0}^{m-1}\sigma(a_{i})\delta(S_{i,k})& \\
&\stackrel{\eqref{Sik+1}}{=} & \sum_{i=0}^{m-1}\sigma(a_{i-1})S_{i,k}\\&&- \sum_{i=0}^{m-1}\sigma(a_i)[\sigma(S_{i,k+1})+S_{i+1,k}] \\
&= & \sum_{i=0}^{m-1}\sigma(a_{i-1})S_{i,k}- \sigma(\sum_{i=0}^{m-1}a_iS_{i,k+1})
\\&&-\sum_{i=0}^{m-1}\sigma(a_{i})S_{i+1,k} \\
&\stackrel{\eqref{sumas}}{=}& \sum_{i=0}^{m-1}\sigma(a_{i-1})S_{i,k}-\sum_{i=0}^{m-1}\sigma(a_{i})S_{i+1,k} \\
&= & -\sigma(a_{m-1})S_{m,k}.  \\
\end{array}$$

Since, by \eqref{mucero}, $\varphi_u^m+\sum_{i=0}^{m-1}\mu_i\varphi_u^i=0$, we get 
\[
\begin{array}{rcl}S_{m,k}&=&\sum_{j=1}^{v}\varphi_u^{m+k_j}(\alpha)\psi^k(e_{k_j})\\&=&\sum_{j=1}^{v}[- \sum_{i=0}^{m-1}\mu_i\varphi_u^{k_j+i}(\alpha)]\psi^k(e_{k_j})\\&=&- \sum_{i=0}^{m-1}\mu_i\sum_{j=1}^{v}\varphi_u^{k_j+i}(\alpha)\psi^k(e_{k_j})\\&=&-\sum_{i=0}^{m-1}\mu_iS_{i,k}.
\end{array}
\] 
Then $\sum_{i=0}^{m-1}(\sigma(a_{i-1})+\delta(a_i))S_{i,k}=\sum_{i=0}^{m-1}\sigma(a_{m-1})\mu_iS_{i,k}$ and, therefore,
\[
(b_0,b_1, \dots, b_{m-1})A_T E_{\theta}=0,
\]
where $b_i = \sigma(a_{i-1}) + \delta(a_i) - \sigma(a_{m-1})$ for $i=0, \dots, m-1$.  

We thus deduce from \eqref{fea} that $x (\sum_{i=0}^{m-1}a_ix^i) \in I$ whenever $\sum_{i=0}^{m-1}a_ix^i \in I$. 
Hence, $I$ is a left ideal of $\mathcal{R}$ and $I = \mathcal{R}\rho$ for some nonzero polynomial $\rho$. As for its degree concerns, we have
\[
\deg \rho = \dim_K \frac{\mathcal{R}}{\mathcal{R}\rho} = \dim_K \frac{K^m}{V} = \theta,
\]
since $A_TE_\theta$ is full rank.

\eqref{rocalculo} We know from \eqref{ro} that, if $\rho = \rho_0 + \cdots + \rho_{\theta}x^\theta$, then the vector $(\rho_0,\dots, \rho_{\theta}, 0, \dots, 0) \in K^m$ belongs to the left kernel of $A_TE_{\theta}$. Now, the statement should be clear. 
\end{proof}

Next, we will state the result that will allow the location of the error positions. We need to construct a matrix from the polynomial $\rho$ given in Proposition \ref{Prop 3}.

For $j = 0, \dots, m-1$  and $i = 0, \dots m-\theta-1$, set 
\[
l_{0,j} = \begin{cases}
\rho_j & \text{if } j= 0, \dots, \theta \\
0 & \text{if } j = \theta + 1, \dots, m-1 \end{cases},  \qquad l_{i,-1} = 0.
 \]
We may then construct a matrix 
\begin{equation}\label{ELE}
L = \left( \begin{array}{ccccc}
l_{0,0} & l_{0,1} && \cdots & l_{0,m- 1} \\
l_{1,0} & l_{1,1} && \cdots & l_{1,m - 1}  \\
\vdots &\vdots && \ddots& \vdots\\
l_{m-\theta-1,0} & l_{m-\theta-1,1} && \cdots &  l_{m-\theta -1,m - 1}
\end{array} \right)
\end{equation}
by defining its entries recursively as
\[
l_{i+1,j} = \sigma(l_{i,j-1}) + \delta(l_{i,j}).
\]

For $i= 0, \dots, m-1$, let $\epsilon_i$ denote the vector of $K^m$ whose $i$--th component is equal to $1$, and every other component is $0$. By $Row(LA)$ we denote the row space of the matrix $LA$. 

\begin{theorem}\label{locatordiv}
 If $T = \{ k_1, \dots, k_v \}$ is the set of error positions, then  
\[
T = \{ k \in \{0, \dots, m-1\} : \epsilon_k \notin Row(LA) \}.
\]
\end{theorem}
\begin{proof}
Let  us first prove that the rows of $L$ form a $K$--basis of $V = \ker (\cdot A_TE_{\theta})$, the left kernel of $A_T E_{\theta}$. According to Proposition \ref{Prop 3}, $\mathfrak{v}(\mathcal{R}\rho) = V$, where $\rho =\sum_{i=0}^{\theta}\rho_i x^{i}$.  

Observe that $\rho, x\rho, \dots, x^{m-1-\theta}\rho$ have different degrees $\theta, \dots, m-1$, so they are $K$--linearly independent in $\mathcal{R}$. Since the dimension of $\mathcal{R}\rho$ is $m-\rho$, we get that they form a basis and, hence, the rows of 
 \[
M_\rho = \left( \begin{array}{c} 
\mathfrak{v}(\rho) \\
\mathfrak{v}(x\rho) \\
\vdots \\
\mathfrak{v}(x^{m - 1 - \theta}\rho)
\end{array} \right) 
\] 
give a basis of $\mathfrak{v}(\mathcal{R}\rho)$. Note that the first row of $M_\rho$ is $\mathfrak{v}(\rho)$. Indeed,  a straightforward computation based on \eqref{torcida} lead to admit that the $j$-th row of $L$ is, precisely, $\mathfrak{v}(x^j\rho)$, for $j=0, \dots, m-1-\theta$. Thus, $L = M_\rho$. 

Let $I$ be denote the identity matrix of size $m \times m$, and denote by $I_T$ the submatrix of $I$ formed by the columns at positions $k_1, \dots, k_v$.  Note that $A_T = A I_T$. This implies that $Row (LA) = \ker (\cdot I_T E_\theta)$. Indeed, we have proved that $Row (L) = \ker (\cdot A_T E_{\theta})$, so that
\[
x \in Row (LA) \Leftrightarrow xA^{-1} \in Row (L) \Leftrightarrow xA^{-1} \in \ker (\cdot A_T E_{\theta}) \Leftrightarrow x \in \ker (\cdot I_T E_\theta).
\] 

Let $i \in \{0,\dots, m-1\}$. If $i \in T$, then $\epsilon_i I_T E_{\theta}$ is the $i$-th row of $E_\theta$, while if $i \notin T$, then $\epsilon_i I_T E_{\theta} = 0$. Since every row of $E_{\theta}$ is non zero, we get that $\epsilon_i \in Row(LA)$ if and only if $i \notin T$. 
\end{proof}

In our decoding algorithm, we need to compute  $\theta$ from the received word $y$. To this end, set
\[
\tau = \lfloor \frac{d-1}{2} \rfloor,
\]
the integer part of $(d-1)/2$.

\begin{lemma}\label{lematheta}
For every $r \geq 1$, define the matrix
\[
S_{r} = \left( \begin{array}{ccccc}
S_{0,0} & S_{0,1} && \cdots & S_{0,r - 1} \\
S_{1,0} & S_{1,1} && \cdots & S_{1,r - 1}  \\
\vdots &\vdots && \ddots& \vdots\\
S_{\tau,0} & S_{\tau,1} && \cdots &  S_{\tau,r - 1}
\end{array} \right).
\]
If $v \leq \tau$, then 
$
\theta =   \max \{ r :   rank \, S_r = r \} . 
$
\end{lemma}
\begin{proof}
Observe that $S_r = ME_r$, where 
\[
M = \left( \begin{array}{ccccc}
\varphi_u^{k_1}(\alpha) &\varphi_u^{k_2}(\alpha)  && \cdots &\varphi_u^{k_v}(\alpha)  \\
\varphi_u^{k_1+1}(\alpha) & \varphi_u^{k_2+1}(\alpha) && \cdots &  \varphi_u^{k_v+1}(\alpha) \\
\vdots &\vdots && \ddots& \vdots\\
\varphi_u^{k_1+\tau}(\alpha) & \varphi_u^{k_2+\tau}(\alpha) && \cdots &  \varphi_u^{k_v+\tau}(\alpha)
\end{array} \right).
\]
Since  $v \leq \tau$, the rank of $M$ is $v$ due to Lemma \ref{Wronskian}. We thus get that $rk \, S_r = rk \, E_r$ for all $r \geq 1$, which gives the desired determination of $\theta$.
\end{proof}

Next, we derive the proofs of Proposition \ref{kernelrho} and Theorem \ref{Main}.

\medskip

\noindent
\textbf{Proof of Proposition \ref{kernelrho}.}
Lemma \ref{lematheta} gives that, whenever $v \leq \tau$,  $$ \max \{ r : rank \, S_r = r \} = \theta =  \max \{ r : rank \, E_r = r \}.$$ By Lemma \ref{calcular},  the matrix $B$ consists of the first $\theta + 1$ rows of $A_TE_\theta$, as in Proposition \ref{Prop 3}.(\ref{rocalculo}). Now, $B$ has rank $\theta$, so that its left kernel is of dimension $1$ as a $K$--vector space. Proposition \ref{rocalculo} guarantees that $(\rho_0, \rho_1, \dots, \rho_\theta)$ belongs to this kernel and $\rho_\theta \neq 0$. 
\medskip

\noindent
\textbf{Proof of Theorem \ref{Main}.}
By Lemma \ref{lematheta}, since we assume that $v \leq \tau$, we get that $$ \max \{ r : rank \, S_r = r \} = \theta =  \max \{ r : rank \, E_r = r \}.$$ Thus, Theorem \ref{locatordiv} is of application to obtain the first part of Theorem \ref{Main}. The second statement is given by Proposition \ref{errorsol}. 

\medskip

We finally state and prove the main result in this paper.

\begin{theorem}
Assume $K$ to be a commutative field and $v \leq \tau = \lfloor (d-1)/ 2\rfloor$. Then Algorithm \ref{algoritmo} correctly computes the error vector. 
\end{theorem}
\begin{proof}
The output of Line 1 is $e = 0$ since $C_{(\varphi_u,\alpha,d)}$ is the kernel of the matrix $H$ in Definition \ref{Codefined}. Line 2 runs whenever $e \neq 0$. In such a case, since we are assuming the the number of errors is $v \leq \tau$, it follows from Proposition \ref{errorsol} that $S_{i,0}\neq 0$ for at least one $0 \leq i \leq \tau$, as the linear system \eqref{errorsol2} has a unique solution. This is to mean that $S_1 \neq 0$ and, henceforth, always under the condition $v \leq \tau$, the number $\theta$ computed in Line $2$ equals $\max \{ r : rank \, S_r = r \}$. Proposition \ref{kernelrho} guarantees the existence of a nonzero vector $\rho$ to be computed in Line 3,  which serves as the initial datum to the calculation of the matrix $L$ in Line 4. Finally, Theorem \ref{Main} assures that the error positions and values computed in Lines 5 and 6 lead to a correct output in Line 7. 
\end{proof}

\section{Skew-differential codes as $(\sigma,\delta)$--codes.}\label{SDmodule}
In Section \ref{fundamentos},  the ring $\mathcal{R}$ was proved (Proposition \ref{isopoly}) to be isomorphic to a factor ring of the skew-polynomial ring $R = K[x;\sigma,\delta]$. Indeed, as we will see later, the codes $C_{(\varphi_u, \alpha, d)}$ are left ideals of $\mathcal{R}$ and, henceforth, they constitute a class of  $(\sigma,\delta)$--codes in the sense of \cite{Boucher/Ulmer:2014}, which enjoys an efficient algebraic decoding algorithm (see Algorithm \ref{algoritmo}). In this section, our aim is to describe precisely how the codes $C_{(\varphi_u,\alpha,d)}$ look like from the perspective of 
the ring $R$, although this view, we think,  is less practical, for our purposes, than our choice in the previous sections, which are independent from the forthcoming material.

Given $\raiz \in K$ we may consider the principal left ideal $R(x-\raiz)$ of $R$ generated by $x-\raiz \in R$. Since $K$ is a subring of $R$ in the obvious way, the factor left $R$--module $R/R(x-\raiz)$ is a left $K$--vector space of dimension $1$. An explicit isomorphism
\begin{equation}\label{simpleK}
R/R(x-\raiz) \cong K
\end{equation}
sends the equivalence class of  $g(x) \in R$ onto its \emph{right evaluation} $g[\raiz] \in K$, defined as the remainder of the left Euclidean division
\begin{equation}\label{root1}
g(x) = q(x)(x-\raiz) + g[\raiz], 
\end{equation}
where $q(x) \in R$ is a suitable polynomial. 

The left $R$--module structure of $R/R(x-\raiz)$ is transferred to $K$ via the isomorphism \eqref{simpleK}, and it leads to a ring homomorphism 
\begin{equation}\label{lambda}
\lambda : R \longrightarrow \lend{}{K}, 
\end{equation}
where $\lend{}{K}$ is still denote the ring of all additive endomorphisms of $K$. Recall that $\lambda$ sends $f \in R$ onto the map defined by left multiplication by $f$ according to the left $R$--module structure of $K$. A straightforward computation shows that $\lambda$ sends $f = \sum_i f_ix^i$ onto $\sum_i f_i\varphi_u^i$, so that it acts as the map $\pi$ from Proposition \ref{isopoly}. Henceforth, the kernel of $\lambda$ equals $R\mu = \mu R$, and $\lambda$ induces the isomorphism $R/R\mu \cong \mathcal{R}$ from Proposition \ref{isopoly}. We are assuming, as in Section \ref{fundamentos}, that $\mu \neq 0$ and its degree is $m$. Recall that a $K$--basis of $\mathcal{R}$ is
$\{1, x, \dots, x^{m-1} \}$ where we are identifying each element of $\mathcal{R}$ with its unique representative in $R$ of degree smaller than $m$. This natural basis of $\mathcal{R}$ leads to the corresponding coordinate isomorphism
\[
\mathfrak{v} : \mathcal{R} \to K^m.
\]

\begin{definition}\label{sigmadeltau}
A $K$--linear code $C \subseteq K^m$ is said to be a \emph{$(\sigma, \delta, \raiz)$--code}  if $\mathfrak{v}^{-1}(C)$ is a left ideal of $\mathcal{R}$. These codes will be referred to as \emph{skew-differential codes}.
\end{definition}

\begin{remark}\label{modulecodes}
In \cite{Boucher/Ulmer:2014}, a \emph{module $(\sigma, \delta)$--code} is defined as submodule of a left module of the form $R/Rf$, for some nonzero skew polynomial $f \in R$. Indeed, their definition is given for $K$ a finite field but, obviously, it makes sense for a general field. From this perspective, the $(\sigma,\delta,\raiz)$--codes  are instances of module $(\sigma,\delta)$--codes, when one sets $f=\mu$, the minimal polynomial of $\varphi_\raiz$ over $K^{\varphi_\raiz}$ (and, hence, $R/Rf = \mathcal{R}$). 
\end{remark}

Every $(\sigma,\delta,u)$--code admits a nice presentation in terms of linear polynomials. Recall that, for any $c \in K^*$, we have the conjugate ${}^c\raiz = \sigma(c)\raiz c^{-1} + \delta(c) c^{-1}.$

\begin{proposition}\label{SDcodesversusmodule}
Every $(\sigma,\delta,u)$--code is of the form $$C = \mathfrak{v}(\mathcal{R}g),$$ where 
\begin{equation}\label{lclm}
g = [x-{}^{c_1}\raiz, \dots, x-{}^{c_k}\raiz]_\ell,
\end{equation}
the least common left multiple in $R$ of $x-{}^{c_1}\raiz, \dots, x-{}^{c_k}\raiz$, for some $c_1, \dots, c_k \in K^*$. 
\end{proposition}
\begin{proof}
Observe that the left $\mathcal{R}$--module $R/R(x - \raiz)$ is simple because it is of dimension $1$ as a left $K$--vector space. By Jacoboson-Bourbaki's Theorem (see \cite[Theorem 4.1]{Sweedler:1975}), $\lambda$ gives a ring isomorphism
$
\mathcal{R} \cong \rend{K^{\varphi_u}}{K},
$
so $\mathcal{R}$ is isomorphic to a full matrix ring with coefficients in $K^{\varphi_u}$.  We know thus that every simple left $\mathcal{R}$--module is  isomorphic to $R/R(x-\raiz)$. This entails (see, e.g., \cite[pp. 40-41]{Gomez:2014}) that every maximal left ideal of $\mathcal{R}$ is of the form $\mathcal{R}(x- {}^c\raiz)$ for a suitable non-zero $c \in K$. Since every left ideal of $\mathcal{R}$ is the intersection of finitely many maximal left ideals, we get the description \eqref{lclm}.  
\end{proof}

Our next aim is to discuss when the representation \eqref{lclm} is irredundant, which will also lead to the computation of a parity-check matrix of the code $C$. 
We will use that  the non-commutative evaluation defined in \eqref{root1} obeys some rules, which are to be recalled. 

Following \cite{Lam/Leroy:1988} define by recursion,  for $a \in K$:

$$N_0(a)=1,$$
$$N_{n+1}(a)=\sigma(N_n(a))a+\delta(N_n(a)).$$

If $g(x)=\sum_i g_ix^i\in K[x;\sigma, \delta]$ then,   by \cite[Lemma 2.4]{Lam/Leroy:1988},  

\begin{equation}\label{root2}
g[a] =\sum_i g_iN_i(a).
\end{equation}

Following \cite{Delenclos/Leroy:2007}, define the \emph{Vandermonde matrix}  
\[
V_{n}(c_1,\dots,c_k)=\left( \begin{array}{ccccc}
1 & 1&& \cdots &1 \\
c_1& c_2 && \cdots &  c_k \\
N_{2}(c_1)& N_{2}(c_2) && \cdots &  N_{2}(c_k)  \\
\vdots &\vdots && \ddots& \vdots\\
N_{n-1}(c_1)& N_{n-1}(c_2) && \cdots &  N_{n-1}(c_k)
\end{array} \right) 
\]

and the \emph{Wronskian matrix}  
\[
W_{n}^{\raiz}(c_1,\dots,c_k)=\left( \begin{array}{ccccc}
c_1 &c_2 && \cdots &c_k \\
\varphi_u(c_1)& \varphi_u(c_2) && \cdots &  \varphi_u(c_k)  \\
\vdots &\vdots && \ddots& \vdots\\
\varphi_u^{n-1}(c_1)& \varphi_u^{n-1}(c_2) && \cdots &  \varphi_u^{n-1}(c_k)
\end{array} \right) 
\]
for each $n \geq 1$.  Then, by \cite[Proposition 4.4]{Delenclos/Leroy:2007},
\begin{equation}  \label{Vd=W}
V_{n}({}^{c_1}\raiz,\dots,{}^{c_k}\raiz)diag(c_1,..,c_k)=W_{n}^{\raiz}(c_1,\dots,c_k),
\end{equation}
where $diag(c_1, \dots, c_k)$ denotes the diagonal matrix built from the list $c_1, \dots, c_k$.

\begin{proposition}\label{DCCdim}
Let $\{ c_1, \dots, c_k \} \subseteq K^*$ be a  linearly independent set over $K^{\varphi_u}$, with $k \leq m-1$, and set  
\[
g = [x-{}^{c_1}\raiz, \dots, x-{}^{c_k}\raiz]_\ell.
\]
Then $deg(g) = k$, $g$ is a right divisor of $\mu$,  and $\mathfrak{v}(\mathcal{R}g)$ is the left kernel of the Wronskian matrix 
\[
W_{m}^{\raiz}(c_1,\dots,c_k).
\]
In other words, $C =  \mathfrak{v}(\mathcal{R}g)$ is a $K$--linear skew-differential code of dimension $m-k$ with parity-check matrix $W_{m}^{\raiz}(c_1,\dots,c_k)$. 
\end{proposition}

\begin{proof} 
By  \cite[Theorem 5.3]{Delenclos/Leroy:2007}, $deg(g) = k$.
On the other hand,   $f=\sum_{k=0}^{m-1}f_kx^k\in \mathcal{R}g$ if and only if $x-{}^{c_j}\raiz$ right divides $f$ for all $j=1,\dots,m$. This is equivalent, by \eqref{root1} and \eqref{root2}, to the condition \[
(f_0,\dots, f_{m-1})V_m({}^{c_1}\raiz,...,{}^{c_k}\raiz)=0.
\]

By \eqref{Vd=W}, this is equivalent to the condition \[(f_0,\dots, f_{m-1})W_m^{\raiz}(c_1,\dots,c_k)=0\]
as required.
\end{proof}

We are now ready to locate our RS skew-differential codes within the class of all $(\sigma,\delta,u)$--codes. 

\begin{corollary}\label{ubicacion}
The code $C_{(\varphi_u, \alpha, d)}$ is a $(\sigma, \delta, u)$--code given by $C_{(\varphi_u, \alpha, d)} = \mathfrak{v}(\mathcal{R}g)$, where
\[
g = [x-{}^\alpha\raiz,x-{}^{\varphi_u(\alpha)}\raiz, \dots, x-{}^{\varphi_u^{d-2}(\alpha)}\raiz]. 
\]
\end{corollary}

\begin{remark}\label{ncgenpol}
The code $C_{(\varphi_u,\alpha,d)}$ was defined  by means of its parity-check matrix $H$ (see Definition \ref{Codefined}). Of course, in order to specify the encoding of messages, one may use the standard method for linear codes of constructing a generator matrix from $H$, as done in Subsections \ref{block} and \ref{convolutional}.

An alternative is to use the arithmetic of the ring $\mathcal{R}$. Indeed, one may compute the skew polynomial $g$ from Corollary \ref{ubicacion} and use it as an encoder similarly to the commutative cyclic case.  As for the computation of $g$ concerns, one may use the non-commutative extended Euclidean algorithm (see, e.g.  \cite[Ch. I, Theorem 4.33]{Bueso/alt:2003}). For instance, in the example described in Subsection \ref{block}, the set of conjugates becomes
$$\{^{\alpha}u , ^{\varphi_u(\alpha)}u , ^{\varphi^2_u(\alpha)}u, ^{\varphi^3_u(\alpha)}u\}=\{a^{137}, a^{212}, a^{141}, a^{225}
\},$$
so that a generator polynomial of this code is
\[
g =  \left [ x-a^{137}, x-a^{212}, x-a^{141}, x-a^{225} \right ]_\ell 
=  x^4+a^{187}x^3+a^{99}x^2+a^{98}x+ a^{218}.
\] 
\end{remark}

\begin{remark}
Module $(\sigma,\delta)$--codes over a finite field $\mathbb{F}$ generated be a polynomial of the form $[x-\alpha_1, \dots, x-\alpha_n]_\ell \in \mathbb{F}[x;\sigma,\delta]$ have been proved to be MDS in \cite[Theorem 5]{Boucher/Ulmer:2014} whenever $\alpha_1, \dots, \alpha_n \in \mathbb{F}$ are suitable powers of an element in the algebraic closure of $\mathbb{F}$ subject to additional conditions called ``Hamming 1'' and ``Hamming 2''. These codes are different, in the case $K = \mathbb{F}$, from that of Corollary \ref{ubicacion}, which are known to be MDS by Theorem \ref{MDS}.  Also, a decoding algorithm, different from Algorithm \ref{algoritmo} was designed in \cite{Boucher/Ulmer:2014}, under the condition ``Hamming 1'' which, in particular, requires $\delta = 0$. 
\end{remark}

Finally, we analyze how skew RS codes from \cite{Gomez/alt:2018a} and  RS differential convolutional codes \cite{Gomez/alt:2019b} are particular examples of RS skew-differential  codes. In this way, Algorithm \ref{algoritmo} both extends to a considerable broader class of codes and, also, simplifies the decoding algorithms designed in \cite{Gomez/alt:2018a} and \cite{Gomez/alt:2019b}. 

\begin{example}\label{RSauto}
Let  $\sigma$ be an automorphism of $K$ of finite order $m$, and choose a cyclic vector $\alpha$ of $\sigma$ as a vector space over $K^{\sigma}$. Set $\beta = \alpha^{-1}\sigma(\alpha)$ and $$g = [x-\beta, x-\sigma(\beta), \dots, x - \sigma^{d-2}(\beta)]_\ell,$$ the left least common multiple being computed in $K[x;\sigma]$.  Skew RS codes from \cite{Gomez/alt:2018a}  are defined as $\mathfrak{v}(\mathcal{R}g$), where $$\mathcal{R} = K[x;\sigma]/\langle x^m - 1 \rangle.$$ Since $1^\alpha = \beta$, and $\varphi_1 = \sigma$ we see, after Corollary \ref{ubicacion}, that $$\mathfrak{v}(\mathcal{R}g) = C_{(\sigma,\alpha,d)}.$$ That is, for $u = 1$ and $\delta = 0$, we obtain the skew RS codes from \cite{Gomez/alt:2018a}. Thus, we may apply the decoding algorithm presented here, which is simpler than that of \cite{Gomez/alt:2018a}. 
\end{example}

\begin{example}\label{RSder}
Let $\delta$ be any $\mathbb{F}$--linear derivation of the the field $\mathbb{F}(t)$ of rational functions in the variable $t$ with coefficients in a finite field $\mathbb{F}$.  If $p$ is the characteristic of $\mathbb{F}$, then  the degree of $\mathbb{F}(t)$ over the field of constants $\mathbb{F}(t)^\delta$ is $p$, and the minimal polynomial of $\delta$ becomes $\mu = x^p - \gamma x$, where $\gamma = \delta^p(t)/\delta(t)$ (see \cite{Gomez/alt:2019b} for details). For any $c \in \mathbb{F}(t)$, the logarithmic derivative is defined as $L(c) = c^{-1}\delta(c)$. Choose a cyclic vector $\alpha$ for the $\mathbb{F}(t)^\delta$--linear map $\delta$ and set $$g = [x-L(\alpha),x- L(\delta(\alpha)), \dots, x- L(\delta^{d-2}(\alpha))]_\ell \in \mathbb{F}(t)[x;\delta].$$ In \cite{Gomez/alt:2019b}, the differential convolutional RS codes are defined as $\mathfrak{v}(\mathcal{R}g)$, where, this time, $$\mathcal{R} = \mathbb{F}(t)[x;\delta]/\langle x^p - \gamma x \rangle. $$ Since $0^\alpha = L(\alpha)$ and $\varphi_0 = \delta$, we deduce from Corollary \ref{ubicacion} that $$\mathfrak{v}(\mathcal{R}g) = C_{(\delta,\alpha,d)}.$$ In other words, we obtain the differential convolutional RS codes from \cite{Gomez/alt:2019b} by setting $\sigma = id_{\mathbb{F}(t)}$ and $u = 0$, to which we also may apply the decoding algorithm presented in this paper. Again, it results simpler than that from \cite{Gomez/alt:2019b}.
\end{example}

\bibliographystyle{amsplain}

\end{document}